\documentclass[runningheads]{llncs}
\usepackage[T1]{fontenc}
\usepackage{graphicx}

\usepackage{multirow}
\usepackage{multicol}
\usepackage{blkarray}
\usepackage{float}
\usepackage{placeins}
\usepackage{pdflscape}
\usepackage{afterpage}

\usepackage{lipsum}

\usepackage{amsfonts}

\usepackage{mathtools}
\usepackage{enumitem}
\usepackage{subcaption}
\usepackage[ruled,vlined,linesnumbered, onelanguage]{algorithm2e}
\DontPrintSemicolon

\usepackage{tikz}
\usepackage{colortbl}
\usepackage{arydshln}
\usepackage{etoolbox}
\usepackage{enumitem}

\usetikzlibrary{patterns.meta,calc}
\usetikzlibrary{shapes,arrows,shadows,positioning}
\usetikzlibrary {shapes.multipart}
\usetikzlibrary{decorations.pathreplacing}
\usetikzlibrary{patterns}
\usetikzlibrary{matrix}
\usetikzlibrary{chains}
\usetikzlibrary{shapes.misc}

 \usepackage{xcolor}

\colorlet{lblue}{blue!50!white}
\colorlet{lred}{red!50!white}
\colorlet{lgreen}{green!50!white}
\colorlet{lpurple}{purple!50!white}
\colorlet{lorange}{orange!50!white}
\colorlet{lpink}{pink!50!white}
\colorlet{lbrown}{brown!50!white}
\colorlet{lyellow}{yellow!50!white}
\colorlet{lolive}{olive!50!white}

\usepackage{hyperref}
\hypersetup{
    colorlinks=true,
    linkcolor=purple,
    citecolor=purple,
    urlcolor=purple,
    filecolor=purple,
}

\usepackage{pifont}

\newcommand{\textnumbering}[2][1]{%
  \let\nodecontents\empty%
  \let\counters\empty%
  \begin{tikzpicture}%
  \foreach \x[count=\currentX from 0] in {#2} {
    \pgfmathtruncatemacro\number{\currentX+#1}
    \expandafter\gappto\expandafter\nodecontents\expandafter{\x\&}
    \expandafter\gappto\expandafter\counters\expandafter{\number\&}
  }
  \matrix[inner sep=0pt, matrix of nodes, ampersand replacement=\&, anchor=base, nodes={inner sep=0pt},
  row 1/.style={font=\tiny}]
  {\counters\\[2pt]
  \nodecontents\\};
\end{tikzpicture}%
}

\newcommand{\ch}[1]{\textnormal{\texttt{#1}}}

\newcommand{\kmer}[1][]{{$k$-mer#1}\xspace}
\newcommand{\mmer}[1][]{{$m$-mer#1}\xspace}
\DeclareMathOperator{\vmin}{\textnormal{vmin}}

\begin{document}
\title{Vigemers: on the number of \kmer{s} sharing the same XOR-based minimizer}
%
%
%
\author{Florian Ingels \and Antoine Limasset \and Camille Marchet \and Mikaël Salson}

\authorrunning{F. Ingels, A. Limasset, C. marchet, M. Salson}

\institute{Univ. Lille, CNRS, Centrale Lille, UMR 9189 CRIStAL, F-59000 Lille, France}

\maketitle              
\begin{abstract}

In bioinformatics, minimizers have become an inescapable method for handling \kmer{s} (words of fixed size $k$) extracted from DNA or RNA sequencing, whether for sampling, storage, querying or partitioning. According to some fixed order on \mmer{s} ($m<k$), the minimizer of a \kmer is defined as its smallest \mmer{} --- and acts as its fingerprint. Although minimizers are widely used for partitioning purposes, there is almost no theoretical work on the quality of the resulting partitions. For instance, it has been known for decades that the lexicographic order empirically leads to highly unbalanced partitions that are unusable in practice, but it was not until very recently that this observation was theoretically substantiated. The rejection of the lexicographic order has led the community to resort to (pseudo)-random orders using hash functions. In this work, we extend the theoretical results relating to the partitions obtained by the lexicographical order, departing from it to a (exponentially) large family of hash functions, namely where the \mmer{s} are XORed against a fixed key. More precisely, provided a key $\gamma$ and a \mmer $w$, we investigate the function that counts how many \kmer{s} admit $w$ as their minimizer (i.e. where $w\oplus\gamma$ is minimal among all \mmer{s} of said \kmer{s}). This number, denoted by $\pi_k^{\gamma}(w)$, represents the maximum size of the bucket associated with $w$, if all possible \kmer{s} were to be seen and partitioned. We adapt the (lexicographical order) method of the literature to our framework and propose combinatorial equations that allow to compute, using dynamic programming, $\pi_k^{\gamma}(w)$ in $O(km^2)$ time and $O(km)$ space.
\end{abstract}


\section{Introduction}
In bioinformatics, technological advances in sequencing drive an exponential growth of the quantity of biological sequencing data available for analysis --- now in the scale of dozens of Petabytes of public data \cite{chikhi2025logan,katz2022sequence}. With this explosive growth also comes the need for always faster algorithms and more efficient data structures. To this end, the minimizers method \cite{schleimer2003winnowing,roberts2004reducing} has established itself as an invaluable partitioning method, used for instance for genome-wide sequence comparisons, for aligning sequencing data, and so on --- we refer the interested reader to the reviews \cite{ndiaye2024less} and \cite{marccais2017improving}. Basically, sequences are represented by a set of \kmer{s}, that is, words of fixed size $k$ (typically, $k=31$ with popular sequencing techniques and in human biology) on an alphabet $\Sigma$; the minimizer of a \kmer is defined at its smallest \mmer (with $m<k$), according to some order over the set of all possible \mmer{s} (e.g. the lexicographical order). Minimizers act as the fingerprints of \kmer{s}. As such, \kmer{s} sharing the same minimizers can be partitioned together, queried and stored in an efficient manner \cite{ingels2025minimizer,martayan2024hyper,li2015mspkmercounter,deorowicz2015kmc}. Another spread use of minimizers is for sampling, i.e. summing up a long sequence by a smaller set of minimizers. Applications include alignment~\cite{li2018minimap2}, sketching \cite{ingels2025compressed,agret2022toward}, or assembly \cite{benoit2024high,benoit2025high}. Crucially, minimizer-based sampling allows to gain order of magnitudes (both in term of speed and disk usage) when estimating the similarity between two sequences-based objects (reads, genomes, etc.), provided they are of good enough quality so as not to lose the relevant signal when sampling. Again, more details can be found in the aforementioned reviews \cite{ndiaye2024less,marccais2017improving}.

Most of the literature dedicated to minimizers on the theoretical aspect has been dedicated to the so-called \emph{density} of minimizer schemes, that is, the expected fraction of distinct minimizers sampled over a random sequence --- we refer the interested reader to the recent thesis of Groot Koerkamp \cite{groot2025optimal}.

However, there is another theoretical approach, which has been little exploited, and which consists of considering how the \kmer{s} are distributed among the \mmer{s}, i.e. for a fixed \mmer $w$, which \kmer{s} admits $w$ as their minimizer. In recent work \cite{ingels2024number}, we studied the counting function $\pi_k^{\text{lex.}}(\cdot)$ of \emph{lexicographical} minimizers, where $\pi^{\text{lex.}}_k(w)$ counts how many \kmer{s} admit $w$ as their lexicographically smaller \mmer. Our goal in this article is to extend this function to other orders on \mmer{s}, so as not to be limited to the lexicographic order alone --- which is rarely used in practice in bioinformatics applications. There are several benefits in knowing this counting function $\pi_k(\cdot)$ for other orders on \mmer{s}:
\begin{itemize}[label=\textbullet]
    \item This allows to anticipate the quality of a minimizer-based partition of \kmer{s}. Indeed, $\pi_k(w)$ counts the maximum size of the bucket associated with $w$, if all \kmer{s} were to be seen and partitioned. Although in real-life applications merely a fraction of all possibles \kmer{s} is observed, we empirically found in \cite{ingels2024number} a real correlation  between this worst-case theoretical value and the empirical filling of buckets for the lexicographical order. For many applications, the goal is to obtain a partition that is as balanced and homogeneous as possible --- which does not happen with the lexicographical order, hence the need to resort to others orders.
    \item Considering $\pi_k(w)/\Sigma^k$ as the probability for a random \kmer to have $w$ as its minimizer, we obtain information on which minimizers are likely to be chosen in a sampling process. Since the similarity between sequences is estimated by the similarity between samples \cite{broder1997resemblance}, identifying biases in the distribution would make it possible to correct accordingly the estimation.
    \item In terms of space usage, encoding the most frequent minimizers with smaller fingerprints become important when dealing with \kmer{s} at a billion scale --- see, for instance \cite{pandey2018mantis}; one can also mention the encoding of so-called canonical \kmer{s} \cite{wittler2023general}.
\end{itemize}

In this article, we propose extending the lexicographical minimizer counting function of \cite{ingels2024number} to an entire family of other orders on \mmer{s}, based on the XOR function and the choice of a key. This allows us to move from covering a single order (lexicographic) to $\Sigma^m$ possible orders. Although still negligible in relation to the total number of $(\Sigma^m)!$ possible orders, we believe that this exponential number of choices nevertheless allows  sufficient expressiveness for applications. In particular, we envision the community using this counting function to evaluate the quality of a given order, depending on the desired use of minimizers and the available data --- and even to carefully craft the best order (among those allowed by the method) for a given task.

\section{Vigemers and vigemins}

\subsection{Preliminaries}\label{ss:prelim}

Let $\Sigma$ be a totally ordered set, called the alphabet, whose elements are called letters. We denote by $a,b,c,\dots$ indeterminate letters, whereas determined letters are denoted in small capitals $\ch{A},\ch{B},\ch{C},\dots$. A word over $\Sigma$ is a finite sequence of letters $w=a_1\cdots a_n$, with $a_i\in\Sigma$. Words are denoted by $w,x,\dots$. The size of $w$, denoted by $|w|$, is equal to $n$. A word of size $k>0$ is called a \kmer, and the set of all \kmer{s} over $\Sigma$ is denoted by $\Sigma^k$. We recall that the lexicographical order over $\Sigma^k$ is defined as follows. Let $x=a_1\cdots a_k$ and $y=b_1\cdots b_k$ be two \kmer{s}; then $x>y$ if and only if either (i) $a_1>b_1$ or (ii) $a_1\cdots a_i = b_1\cdots b_i$ for some  $1\leq i\leq k-1$ and $a_{i+1}>b_{i+1}$. We denote by $\varepsilon$ the empty word, so that $\varepsilon < x$ for any \kmer $x$; in particular, $\varepsilon < a$ for any letter $a\in \Sigma$. Finally, for any two words $x=a_1\cdots a_m$ and $y=b_1\cdots b_k$, with $k\geq m$, we denote by $x\subseteq y$ the fact that $x$ is a subword of $y$ --- i.e. $\exists j\geq 0$ so that $b_{j+1} = a_1$, $b_{j+2}=a_2$, etc.

Whereas the work \cite{ingels2024number} we build upon allowed any alphabet $\Sigma$, we restrict ourselves here to alphabets where $|\Sigma|=2^b$ for some $b\geq 0$. In particular, the DNA alphabet $\Sigma = \lbrace \ch{A}, \ch{C},\ch{G},\ch{T}\rbrace$, used in bioinformatics, falls into this category. We associate each letter of $\Sigma$ to a $b$-bit vector, and for any two letters $a,b\in\Sigma$, we define the XOR operator $a\oplus b$ as the letter in $\Sigma$ whose bit-vector is obtained by taking the XOR of the bitvectors of $a$ and $b$, component by component. 

\vspace{0.25\baselineskip}
\noindent
\begin{minipage}[c]{0.6\textwidth}
In the case of the DNA alphabet, we get Table~\ref{tab:xor_dna} --- matching \ch{A} with \texttt{00}, \ch{C} with \texttt{01}, \ch{G} with \texttt{10} and \ch{T} with \texttt{11}. We retrieve the famous Klein group $(\mathbb{Z}/2\mathbb{Z})^2$, whereas for a general alphabet we would retrieve the elementary abelian group $(\mathbb{Z}/2\mathbb{Z})^b$, sometimes called boolean groups \cite{halmos2009introduction}. 
\end{minipage}\hfill
\begin{minipage}[c]{0.15\textwidth}
\centering
\begin{tabular}{c|cccc}
    $\oplus$ & \ch{A} & \ch{C} & \ch{G} & \ch{T}\\
    \hline
   \ch{A}  &  \ch{A} & \ch{C} & \ch{G} & \ch{T} \\
   \ch{C} & \ch{C}& \ch{A} & \ch{T} &\ch{G}\\
   \ch{G} & \ch{G}& \ch{T}& \ch{A} & \ch{C} \\
   \ch{T} & \ch{T} & \ch{G} & \ch{C} & \ch{A} \\
\end{tabular}
\end{minipage}\hfill
\begin{minipage}{0.2\textwidth}
\captionof{table}{Cayley table of $\oplus$ using DNA alphabet.}\label{tab:xor_dna}
\end{minipage}

Then, for any two \mmer{s} $x=a_1\cdots a_m$ and $y = b_1\cdots b_m$, we define $x\oplus y$ as the \mmer $(a_1\oplus b_1)\cdots (a_m\oplus b_m)$. Finally, note that $\oplus$ is an involution --- i.e. $(x\oplus y) \oplus y = x \oplus (y\oplus y) = x$. By convention, $\varepsilon\oplus a = \varepsilon$ for any $a\in\Sigma \cup \lbrace \varepsilon\rbrace$.

Let $\gamma \in\Sigma^m$ be a key. For any \kmer $x\in\Sigma^k$, we denote by $V_\gamma(x)$ its set of \emph{vigemers}\footnote{This is a pun/reference to Vigenère cipher \cite{bruen2011cryptography}, a 16th century encryption technique very close to the concept at hand.}, that is, the set 
$V_\gamma(x) = \lbrace w\oplus \gamma : (w\in\Sigma^m) \wedge (w\subseteq x)\rbrace$
of the \mmer{s} of $x$ XORed with $\gamma$.  

\begin{definition}[Vigemin]\label{def:vigemin}
Let $\gamma\in\Sigma^m$ be a key. For any \kmer $x$, the \emph{vigemin} of $x$, denoted by $\vmin_\gamma(x)$, is the leftmost \mmer $w$ of $x$ such that the word $w\oplus \gamma$ is minimal among the set $V_\gamma(x)$, for the standard lexicographical order.
\end{definition}

Note that, when $\gamma = \ch{A}\cdots\ch{A}$, $V_\gamma(x)$ denotes simply the set of \mmer{s} of $x$, and $\vmin_\gamma(x)$ the lexicographical minimizer of $x$ --- where ties are standardly resolved to the left \cite{schleimer2003winnowing}. Some of the usual variants of lexicographic minimizers can also be expressed in the vigemer framework, demonstrating its broad relevance, such as the \emph{alternating order} \cite{roberts2004reducing} --- where \mmer{s} are compared using the lexicographical order for even positions, and the reverse lexicographical order for odd positions --- obtained with $\gamma = \ch{A}\ch{T}\ch{A}\ch{T}\cdots$; or the \emph{anti-lexicographical order}, proposed in \cite[Def.~8.4]{groot2025optimal} --- where \mmer{s} are compared using the lexicographical order for their first character, and the reverse lexicographical order for subsequent characters --- obtained with $\gamma = \ch{A}\ch{T}\cdots\ch{T}$.

Traditionally, minimizers are defined using an order $\mathcal{O}_m$ on \mmer{s}. $\mathcal{O}_m$ is an injective function $\Sigma^m\to \mathbb{R}$ so that $x\leq_{\mathcal{O}_m}y$ if and only if $\mathcal{O}_m(x)\leq \mathcal{O}_m(y)$. The minimizer of a \kmer is its minimal \mmer for the order $\mathcal{O}_m$ \cite{schleimer2003winnowing,roberts2004reducing,martayan2024hyper}. Here, note that the function $w \in \Sigma^m \mapsto \texttt{rank}(w\oplus \gamma) \in [\![1,\|\Sigma|^m]\!]$ does indeed define an order on \mmer{s}, where $\texttt{rank}(w)$ designates the rank of a \mmer $w$ for the lexicographical order. Therefore, our concept of vigemin and vigemers actually encompasses $\Sigma^m$ possible orders on \mmer{s}, one per possible key $\gamma$.

With regard to the introduction, in this article  we are interested in computing the following counting function, for any $w\in\Sigma^m$ and any $\gamma\in\Sigma^m$:
$$\pi_k^\gamma(w) = |\lbrace x\in\Sigma^k : \vmin_\gamma(x) = w \rbrace|.$$

In other words, $\pi_k^\gamma(w)$ counts how many \kmer{s} admit $w$ as their vigemin. In \cite{ingels2024number}, we computed $\pi_k^{\ch{A}\cdots \ch{A}}(\cdot)$, by establishing a set of equations, solved using dynamic programming. In this work, we adapt this previous method to deal with any possible $\gamma$. Although the calculations will essentially follow the same steps and involve the same concepts, generalization is not trivial. For the sake of self-containedness, we will reintroduce all necessary concepts from \cite{ingels2024number}, with due credit when no or only minor adaptation is required.

For the rest of the paper, we assume that $k,m,w=a_1\cdots a_m$ and $\gamma=c_1\cdots c_m$ are fixed, and any dependence on them will be assumed implicitly as to simplify the notations. Finally, Iverson brackets \cite{iverson1962programming} are used to denote indicator functions; that is $[P] =1$ if property $P$ is true, and $[P]=0$ otherwise.

\subsection{Autocorrelation matrix and specialized alphabets}

The main tool to compute $\pi_k^\gamma(w)$ is called the \emph{autocorrelation matrix}, as a generalization of the classical autocorrelation vector \cite{rivals2003combinatorics}. As we shall see later, a large part of the calculation of $\pi_k^\gamma(w)$ involves comparing substrings of $w$ against prefixes of $w$, XORed by the key $\gamma$. In particular, we are interested in the sign of these comparisons: $=$, $>$, or $<$. 

\begin{definition}[Autocorrelation matrix, \cite{ingels2024number}]
The \emph{autocorrelation matrix} of $w=a_1\cdots a_m$ XORed by $\gamma= c_1\cdots c_m$ is defined as the lower triangular matrix denoted by $\mathbf{R}$, where $\mathbf{R}_{i,j} \in \lbrace <,=,>\rbrace$, with $1\leq j\leq i\leq m$, so that 
$$\big((a_j\cdots a_i)\oplus(c_1\cdots c_{i-j+1})\big) \mathrel{\mathbf{R}_{i,j}} \big((a_1\cdots a_{i-j+1})\oplus (c_1\cdots c_{i-j+1})\big).$$
\end{definition}

To alleviate notations, instead of writing $\mathbf{R}_{i,j} = \star$, with $\star \in \lbrace <,=,>\rbrace$, we define the binary variables $\mathbf{R}_{i,j}^\star= [\mathbf{R}_{i,j} = \star]$, that we use interchangeably as a binary number or as a logical true/false.

During computations, we construct the \kmer{s} whose vigemin is $w$, letter by letter. Depending on the context, not all letters from $\Sigma$ will be available. We define what we call the \emph{specialized alphabets}, of which there are $m$.

\begin{definition}[Specialized alphabet]\label{def:sigma_i}
For $1\leq i\leq m$, we define the $i$-th \emph{specialized alphabet}, denoted by $\Sigma_i$, as : $\Sigma_i = \lbrace a\in\Sigma : (a\oplus c_i) > (a_i\oplus c_i)\rbrace$.
\end{definition}

\subsection{Antemers and postmers}

We start with the following two definitions.

\begin{definition}[Antemers, \cite{ingels2024number}]
A word $y$ is an \emph{antemer} if and only if, for all \mmer{s} $w'$ (but the last one) of the word $yw$, $w'\oplus\gamma> w\oplus\gamma$. We denote by $A(\alpha)$ the number of antemers of size $\alpha$, and by $A_i(\alpha)$, $0\leq i\leq m$, the number of antemers of size $\alpha$ that share with $w$ a prefix of size $i$.
\end{definition}
By definition, $A_m(\alpha)=0$, $A_i(\alpha)=0$ for $i>\alpha$. We have $A(0)=1$ and also $A(\alpha) = \sum_{i=0}^{m-1}A_i(\alpha)$.

\begin{definition}[Postmers, \cite{ingels2024number}]
A word $z$ is a \emph{postmer} if and only if, for all \mmer{s} $w'$ of $z$, $w'\oplus \gamma\geq w\oplus\gamma$. We denote by $P(\beta)$ the number of postmers of size $\beta$, and by $P_i(\beta)$, $0\leq i\leq m$, the number of postmers of size $\beta$ that share with $w$ a prefix of size $i$.
\end{definition}
Again, $P_i(\beta)=0$ for $i>\beta$, $P(0)=1$ and $P(\beta)=\sum_{i=0}^m P_i(\beta)$.

Now, consider a \kmer $x\in\Sigma^k$ whose vigemin is $w$. There exist $\alpha,\beta\geq 0$, such that $\alpha+\beta+m=k$, and two (possibly empty) words $y\in\Sigma^\alpha$ and $z\in\Sigma^\beta$ such that $x=ywz$. Following Definition~\ref{def:vigemin}, and since ties are resolved to the left, then (1) $y$ must be an antemer, and (2) $wz$ must be a postmer (sharing with $w$ a prefix of size $m$). Moreover, since there is no \mmer of $x=ywz$ that possesses letters both from $y$ and $z$, the two words $y$ and $z$ are independent from each other. Therefore,
\begin{equation}\label{eq:reformulation}
\pi_k^\gamma(w) = \sum_{\alpha+\beta = k-m} A(\alpha)\cdot P_m(\beta+m). 
\end{equation}

We can refine this formula a little further by noticing that certain configurations are forbidden. 

\begin{lemma}[\cite{ingels2024number}]
If there exists $2\leq j\leq i \leq m$ such that $\mathbf{R}_{i,j}^<$, then (i) $A_i(\alpha) = 0$ for $\alpha\geq i$ and (ii)  $P_i(\beta)=0$ for $\beta \geq m+j-1$.
\end{lemma}
\begin{proof}
The proof is straightforwardly adapted from \cite[Prop.~1]{ingels2024number}. Let $i,j$ be so that $\mathbf{R}_{i,j}^<$, and consider the word $x = a_1\cdots a_j\cdots  a_i b_{i+1}\cdots b_p$. We want to forbid $x$ to have a \mmer $w'$ so that $w'\oplus \gamma < w\oplus \gamma$ (in which case, $w$ could never be the vigemin of $x$). Since $(a_j\cdots a_i)\oplus (c_1\cdots c_{i-j+1}) < (a_1\cdots a_{i-j+1})\oplus (c_1\cdots c_{i-j+1})$, if there are enough letters left in $x$ to complete $a_j\cdots a_i$ into a \mmer, then we would break our requirement no matter which letters are chosen. So, we must have $(p-i) + (i-j+1)<m$, i.e. $p<m+j-1$. For postmers, we use $p=\beta$. For antemers, $p=m+\alpha$; remember that $A_i(\alpha)=0$ if $\alpha <i$, therefore $\alpha \geq i\geq j-1$ leads to $p=\alpha+m\geq m+j-1$, hence the result.
\end{proof}

Then, we can define $i_{\max}$ and $\beta_{\max}$ as the maximum values that $i$ (resp. $\beta$) can take without $A_i(\alpha)$ (resp. $P_{m}(\beta+m)$) being always zero:
\begin{align}\label{eq:i_beta_max}
\begin{split}
    i_{\max} &= \min\left( \lbrace m\rbrace \cup \lbrace 2\leq i\leq m-1: \exists 2\leq j\leq i : \mathbf{R}_{i,j}^<\rbrace\right)-1\\
\beta_{\max} &= \min\left(\lbrace k-m+2\rbrace \cup \lbrace 2\leq j\leq m : \mathbf{R}_{m,j}^<\rbrace\right) -2
\end{split}
\end{align}

Finally, combining with \eqref{eq:reformulation} we end up with
\begin{proposition}[\cite{ingels2024number}]\label{prop:reformulation}
$$\pi_k^\gamma(w) = \sum_{\beta=0}^{\beta_{\max}} A(k-m-\beta)\cdot P_m(\beta+m)$$
where $A(\alpha) = \sum_{i=0}^{i_{\max}} A_i(\alpha)$ and $P(\beta) = \sum_{i=0}^m P_i(\beta)$, $A(0)=P(0)=1$ and $A_i(\alpha) = P_i(\beta)=0$ whenever $i>\alpha,\beta$.
\end{proposition}
Then, all that is left is to compute $A(\cdot), A_i(\cdot), P(\cdot)$ and $P_i(\cdot)$ by establishing systems of recurrent equations that link them.

\section{Computing the number of antemers}
In this section, we establish a relation between the different values of $A(\alpha)$ and $A_i(\alpha)$. Remember that $A(\alpha) = \sum_{i=0}^{i_{\max}}A_i(\alpha)$. Since $i\leq i_{\max}$, it means for any $1\leq j\leq i$, either $\mathbf{R}_{i,j}^>$ or $\mathbf{R}_{i,j}^=$. Throughout this section, we denote by $y=b_1\cdots b_\alpha$ a candidate antemer, and suppose that $y$ share with $w$ a prefix of size $i$ --- so that $b_1\cdots b_i= a_1\cdots a_i$. Our goal is to find conditions on the subsequent letter $b_{i+1}$ of $y$. By definition, we seek that all \mmer{s} $w'$ (but the last one) of the word $yw$ are such that $w'\oplus\gamma > w\oplus \gamma$.

\subsection{The cases $i=0$ and $i=\alpha$}

\paragraph{Case $i=0$.} Let $y = b_1\cdots b_\alpha$, with $b_1\neq a_1$. Let $w'$ be the \mmer starting at $b_1$. Then, we must have $w'\oplus \gamma>w\oplus \gamma$. It could resolve into either (i) $b_1\oplus c_1 = a_1\oplus c_1$ and $(b_2\cdots)\oplus (c_2\cdots)>(a_2\cdots)\oplus (c_2\dots)$, or (ii) $b_1\oplus c_1 > a_1\oplus c_1$. Since $\oplus$ is an involution, (i) would imply $b_1=a_1$ which is forbidden, so we must have (ii), and $b_1\in \Sigma_1$ --- remember Definition~\ref{def:sigma_i}. Therefore, 

\begin{lemma}\label{lemma:antemer_i=0}
$A_0(\alpha) = |\Sigma_1|\cdot A(\alpha-1)$.
\end{lemma}

\paragraph{Case $i=\alpha$.} There is at most one possible antemer, namely $y = a_1\cdots a_i$, which exists depending on whether $w$ is the vigemin of the word $yw$. By looking at its \mmer{s}, for $1\leq j\leq i$, one must have
$$(a_j\cdots a_i a_1\cdots a_{m-i+j-1})\oplus (c_1\cdots c_m) > (a_1\cdots a_m)\oplus (c_1\cdots c_m)$$
If $\mathbf{R}_{i,j}^>$, this is trivially verified; otherwise, if $\mathbf{R}_{i,j}^=$, then this is equivalent to
$$(a_1\cdots a_{m-i+j-1})\oplus (c_{i-j+2}\cdots c_m) > (a_{i-j+2}\cdots a_m) \oplus (c_{i-j+2}\cdots c_m)$$
which is either true or false. For $0\leq l\leq i_{\max}-1$, let 
$$\mathbf{S}_l = \Big[\big((a_1\cdots a_{m-l-1})\oplus (c_{l+2}\cdots c_m)\big) > \big((a_{l+2}\cdots a_m) \oplus (c_{l+2}\cdots c_m)\big)\Big]$$
Then we have the following result:
\begin{lemma}\label{lemma:antemer_i=alpha}
$A_i(i) = \displaystyle\prod_{j=1}^i (\mathbf{R}_{i,j}^> + \mathbf{R}_{i,j}^=\cdot \mathbf{S}_{i-j})$.
\end{lemma}

\subsection{General case}

Let $0<i < \min(\alpha,i_{\max}+1)$. Let $y = a_1\cdots a_ib_{i+1}\cdots b_{\alpha}$, with $b_{i+1}\neq a_{i+1}$. By enumerating the \mmer{s} of $yw$ that contain $b_{i+1}$, we obtain the following set of equations
\begin{align}
b_{i+1}&\neq a_{i+1} \tag{prefix}\label{eq:antemers_prefix}\\
\forall 1\leq j \leq i,\quad (a_j\cdots a_ib_{i+1} \cdots)\oplus(c_1\cdots) &> (a_1\cdots)\oplus (c_1\cdots) \tag{$i,j$}\label{eq:antemers:i_j}\\
(b_{i+1}\cdots)\oplus (c_1\cdots) &> (a_1\cdots)\oplus (c_1\cdots)  \tag{$i$}\label{eq:antemers_i}
\end{align}
Note that since $|yw|=\alpha+m$, the $j+1$ \mmer{s} considered here do indeed exist since $j+1\leq i_{\max}+1\leq m$. \eqref{eq:antemers_i} is resolved into $b_{i+1}\oplus c_1 \geq a_1\oplus c_1$. Concerning \eqref{eq:antemers:i_j}, either $\mathbf{R}_{i,j}^>$ and no further condition on $b_{i+1}$ is imposed, or $\mathbf{R}_{i,j}^=$ and we obtain $b_{i+1} \oplus c_{i-j+2} \geq a_{i-j+2}\oplus c_{i-j+2}$. In particular, observe that $\mathbf{R}_{i,1}^=$ is always verified, hence equation $(i,1)$, combined with \eqref{eq:antemers_prefix} amounts to $b_{i+1}\oplus c_{i+1}> a_{i+1}\oplus c_{i+1}$. Since all of those conditions must be met simultaneously, we obtain the following result.

\begin{lemma}\label{lemma:antemers:letters_possible}
For $0<i < \min(\alpha,i_{\max}+1)$, $b_{i+1}\in \Sigma_A(i)$, where $$\Sigma_A(i) = \Sigma_{i+1} \cap \big(\Sigma_1 \cup \lbrace a_1\rbrace\big)\cap \bigcap_{\substack{2\leq j\leq i\\ \mathbf{R}_{i,j}^=}} \big(\Sigma_{i-j+2}\cup\lbrace a_{i-j+2}\rbrace\big).$$
\end{lemma}
We must now consider which recursive case to consider when choosing $b_{i+1}$. If $b_{i+1}$ is chosen so that it prevents any prefix of $w$, then we next consider antemers of size $\alpha - (i+1)$. If $b_{i+1}=a_1$, we start a new prefix with $w$ of size at least $1$ --- therefore we recursively consider antemers of size $\alpha - i$ with a common prefix $\geq 1$. Similarly, when $\mathbf{R}_{i,j}^=$, $a_j\cdots a_i$ is a prefix of $w$, namely $a_1\cdots a_{i-j+1}$; therefore, by choosing $b_{i+1}=a_{i-j+2}$, we extend the prefix and must recursively consider antemers of size $\alpha - j+1$ with a common prefix with $w$ of size at least $i-j+2$ (i.e. antemers starting by $a_j\cdots a_i a_{i-j+2} = a_1\cdots a_{i-j+2}$). One crucial point here is to remark that if there exists $j'> j$ so that $\mathbf{R}_{i,j}^=$, $\mathbf{R}_{i,j'}^=$ and $a_{i-j+2}=a_{i-j'+2}= a$, then by choosing $b_{i+1}=a$ we simultaneously extend two prefixes, but we must recursively consider only the longest (here, $i-j+2$). Indeed, the antemers would look like this:
$$\overbrace{a_j\cdots a_{j'-1}\underbrace{a_{j'}\cdots a_i a}_{|\text{prefix}| \geq i-j'+2}\cdots}^{|\text{prefix}| \geq i-j+2}\cdots.$$
To take into account these cases, we use the following notion of \emph{prefix-letter vectors} of $w$, that is:
\begin{definition}[Prefix-letter vector,~\cite{ingels2024number}]\label{def:prefix_letter_vectors}
For $1\leq i \leq m$, the $i$-th \emph{prefix-letter vector} $\mathbf{T}_i$ of $w$ is defined as a vector in $[\![0,m+1]\!]^{|\Sigma|}$ where, for any $a\in\Sigma$:
\begin{itemize}
    \item $\mathbf{T}_1(a) = 2\cdot[a=a_1]$
    \item for $2\leq i \leq m$, 
    $$\mathbf{T}_i(a) = \begin{cases}
        \min \lbrace 2\leq j \leq i : \mathbf{R}_{i,j}^= \wedge (a_{i-j+2}=a)\rbrace &\text{if this set is not empty};\\
        (i+1) \cdot [a=a_1] &\text{otherwise}.\end{cases}$$
\end{itemize}
\end{definition}
The definition is such that, by choosing $b_{i+1}=a$,
\begin{itemize}
    \item either $\mathbf{T}_i(a)=0$ and we recursively consider an antemer of size $\alpha - (i+1)$;
    \item either $\mathbf{T}_i(a)\neq 0$ and we recursively consider an antemer of size $\alpha - \mathbf{T}_i(a)+1$ sharing with $w$ a prefix of size $\geq i-\mathbf{T}_i(a)+2$.
\end{itemize}
Finally, we end up with the following result.

\begin{proposition}\label{prop:antemer:general_case}
Let $0<i < \min(\alpha,i_{\max}+1)$. Let $\Sigma_A^{=0}(i) = \lbrace a \in \Sigma_A(i) : \mathbf{T}_i(a)=0\rbrace$ and $\Sigma_A^{\neq 0}(i) = \lbrace a \in \Sigma_A(i) : \mathbf{T}_i(a)\neq 0\rbrace$. Then,
$$A_i(\alpha) = |\Sigma_A^{=0}(i)| \cdot A\big(\alpha-(i+1)\big) + \sum_{a\in \Sigma_A^{\neq 0}(i)}\sum_{i' = i - \mathbf{T}_i(a)+2}^{i_{\max}} A_{i'}(\alpha - \mathbf{T}_i(a)+1).$$
\end{proposition}

\section{Computing the number of postmers}
In this section, we establish a relation between $P(\beta)$ (the number of postmers of size $\beta$) and $P_i(\beta)$ (the number of said postmers that also share with $w$ a prefix of size $i$). As stated before, we have $P(\beta) = \sum_{i=0}^{m}P_i(\beta)$. Since  $\beta \leq \beta_{\max}$, we must only consider the cases $\mathbf{R}_{i,j}^>$ or $\mathbf{R}_{i,j}^=$, for $1\leq j\leq i$. We denote by $z=b_1\cdots b_\beta$ a candidate postmer, and suppose that $z$ share with $w$ a prefix of size $i$ --- so that $b_1\cdots b_i= a_1\cdots a_i$. Our goal is, once again, to find conditions on the subsequent letter $b_{i+1}$ of $z$. By definition, we seek that all \mmer{s} $w'$ of $z$ are such that $w'\oplus\gamma > w\oplus \gamma$.

\subsection{Edge cases}

\paragraph{Case $0<\beta \leq m-1$.} Since a postmer of size $<m$ has no \mmer, there is no specific condition on its constituent letters. Therefore, $P(\beta) = |\Sigma|^\beta$. However, we still must compute the intermediate values $P_0(\beta),\dots, P_m(\beta)$ for subsequent equations. We did it in \cite{ingels2024number} and since there is no conditions on $z$ here, the computation remain unchanged. We reproduce in Lemma~\ref{lemma:postmers:small_values} the equations for the sake of self-containedness but we invite the interested reader to refer to \cite[Section~4.1]{ingels2024number} for the proof.

\begin{lemma}[Prop.~7,~\cite{ingels2024number}]\label{lemma:postmers:small_values}
For $0<\beta \leq m-1$, we have $P_0(\beta) = (|\Sigma|-1)\cdot P(\beta-1)$, and, for $0<i<\beta$,
$$P_i(\beta) = |\Sigma_{D}^{=0}(i)|\cdot P\big(\beta-(i+1)\big) + \sum_{a\in \Sigma_{D}^{\neq 0}(i)}\sum_{i' = i - \mathbf{T}_i(a)+2}^{\beta-\mathbf{T}_i(a)+1} P_{i'}(\beta + \mathbf{T}_i(a)+1)$$
where $\Sigma_{D}^{=0}(i) = \lbrace a\in \Sigma : (a\neq a_{i+1})\wedge (\mathbf{T}_i(a)=0)\rbrace$ and $\Sigma_{D}^{\neq0}(i)= \lbrace a\in \Sigma : (a\neq a_{i+1})\wedge (\mathbf{T}_i(a)\neq 0)\rbrace$. Finally, $P_\beta(\beta)=1$.
\end{lemma}

\paragraph{Case $\beta =m$.} A postmer $z = b_1\cdots b_m$ contains a unique \mmer, so we must only check whether $z\oplus\gamma \geq w\oplus\gamma$. Let $z$ share a prefix of size $0\leq i \leq m-1$ with $w$; therefore $b_{i+1}\neq a_{i+1}$. We must have
$$(a_1\cdots a_i b_{i+1}\cdots ) \oplus (c_1\cdots) \geq (a_1\cdots)\oplus (c_1\cdots)$$
which resolve in $b_{i+1}\oplus c_{i+1} > a_{i+1}\oplus c_{i+1}$ (since $b_{i+1}\neq a_{i+1}$). Therefore, $b_{i+1}\in \Sigma_{i+1}$ and the choice of the remaining letters $b_{i+2},\dots$ is free. Hence, we obtain:
\begin{lemma}\label{lemma:postmers:m}
For any $1\leq i \leq m-1$, $P_i(m) = |\Sigma_{i+1}|\cdot |\Sigma|^{m-(i+1)}$. Moreover, $P_m(m)=1$.
\end{lemma}

\paragraph{Case $\beta >m$ and $i=0$.} We have $z=b_1\cdots b_\beta$ with $b_1\neq a_1$. The condition on the first \mmer of $z$ amounts to $(b_1\cdots b_m) \oplus (c_1\cdots c_m) \geq (a_1\cdots a_m)\oplus (c_1\cdots c_m)$, which simplifies into $b_1\oplus c_1 > a_1\oplus c_1$ since $b_1\neq a_1$.

\begin{lemma}\label{lemma:postmers:i_0}
For $\beta >m$, we have $P_0(\beta) = |\Sigma_{1}|\cdot P(\beta-1)$.
\end{lemma}

\subsection{General case}
We are now considering the case $\beta >m$ and $1\leq i \leq m$. Let $z=a_1\cdots a_ib_{i+1}\cdots b_\beta$, with $b_{i+1}\neq a_{i+1}$ --- by convention, $a_{m+1}=\varepsilon$ and $c_{m+1}=\varepsilon$, for the case $i=m$. By examing the \mmer{s} of $z$, we obtain nearly the same set of equations as for antemers, except that depending on $\beta$, some \mmer{s} may not exist. Namely, we retain \eqref{eq:antemers_prefix}, and then:
\begin{align}
\begin{aligned}
\forall 1\leq j \leq i, \text{ iff } \beta\geq m-1+j,\notag\\
(a_j\cdots a_ib_{i+1} \cdots b_{m-1+j})\oplus(c_1\cdots) &> (a_1\cdots)\oplus (c_1\cdots)
\end{aligned}\tag{$i,j,\beta$}\label{eq:postmers:i_j}\\
\begin{aligned}
\phantom{\forall 1\leq j \leq i,}\text{ iff } \beta\geq m+i,\notag\\
(b_{i+1}\cdots b_{m+i})\oplus (c_1\cdots) &> (a_1\cdots)\oplus (c_1\cdots)  
\end{aligned}\tag{$i,\beta$}\label{eq:postmers_i_beta}
\end{align}
Solving these equations leads to the following conditions on $b_{i+1}$:
\begin{itemize}
    \item by \eqref{eq:antemers_prefix} and $(i,1,\beta)$ (since $\beta \geq m$), $b_{i+1}\in \Sigma_{i+1}$ --- where $\Sigma_{m+1}=\Sigma$;
    \item by \eqref{eq:postmers:i_j}, $b_{i+1}\in \Sigma_1\cup \lbrace a_1\rbrace$ if and only if $\beta\geq m+i$;
    \item by \eqref{eq:postmers:i_j}, with $j\geq 2$, $b_{i+1}\in \Sigma_{i-j+2} \cup\lbrace a_{i-j+2}\rbrace $ if and only if $\mathbf{R}_{i,j}^=$ and $\beta\geq m-1+j$.
\end{itemize}
Note that, by definition of $\mathbf{T}_i(a)$ --- recall Definition~\ref{def:prefix_letter_vectors}, the condition $b_{i+1}\in \Sigma_l \cup \lbrace a_l\rbrace$ is equivalent to $\mathbf{T}_i(a_l)\neq 0$. Therefore, we can incorporate the condition on $\beta$ by defining:
\begin{equation}\label{eq:tilde_prefix}
\widetilde{\mathbf{T}}_i(a,\beta) = \mathbf{T}_i(a) \cdot [ \beta\geq m+\mathbf{T}_i(a) -1]
\end{equation}
and then we get the following result.
\begin{lemma}\label{lemma:postmers:letters_possible}
For $\beta >m$ and $1\leq i\leq m$, $b_{i+1}\in \Sigma_P(i,\beta)$, where $$\Sigma_P(i,\beta) = \Sigma_{i+1} \cap \bigcap_{l\in L_i(\beta)} \big(\Sigma_{l}\cup\lbrace a_{l}\rbrace\big)$$
and $L_i(\beta) = \left\lbrace l \in \lbrace 1\rbrace \cup \lbrace i-j+2 : (2\leq j\leq m) \wedge \mathbf{R}_{i,j}^=\rbrace : \widetilde{\mathbf{T}}_i(a_l,\beta)\neq 0\right\rbrace$
\end{lemma}
As for antemers, we partition the alphabet $\Sigma_P(i,\beta)$ into $\Sigma_P^{=0}(i,\beta)$ and $\Sigma_P^{\neq 0}(i,\beta)$ whether the letters $a\in \Sigma_P(i,\beta)$ verify $\widetilde{\mathbf{T}}_i(a_l,\beta)= 0$ or $\neq 0$ (respectively). Then, following a similar line of reasoning to that used for antemers, we end up with:
\begin{proposition}\label{prop:postmers:general_case}
For $\beta >m$ and $1\leq i\leq m$,
$$P_i(\beta) = |\Sigma_P^{=0}(i,\beta)| \cdot P\big(\beta-(i+1)\big) + \sum_{a\in \Sigma_P^{\neq 0}(i,\beta)}\sum_{i' = i -\widetilde{\mathbf{T}}_i(a,\beta)+2}^{m} P_{i'}(\beta - \widetilde{\mathbf{T}}_i(a,\beta)+1).$$
\end{proposition}

\section{Computation complexity and numerical results}\label{sec:computation}
Recall from Proposition~\ref{prop:reformulation} that
$$\pi_k^\gamma(w) = \sum_{\beta=0}^{\beta_{\max}} A(k-m-\beta)\cdot P_m(\beta+m).$$
From previous sections, computing $A(\alpha)$ and $P(\beta)$ involves computing all values of $A_i(\alpha')$ and $P_i(\beta')$ for $0\leq i \leq m$ and $\alpha'\leq \alpha$, $\beta'\leq \beta$. It can be done by dynamic programming by filling a table of dimensions $(\alpha,i_{\max}+1)$ for $A$, and $(\beta,m+1)$ for $P$, and using the various equations found in this article --- namely, Lemmas~\ref{lemma:antemer_i=0} and~\ref{lemma:antemer_i=alpha} and Proposition~\ref{prop:antemer:general_case} for antemers; Lemmas~\ref{lemma:postmers:small_values},~\ref{lemma:postmers:m} and~\ref{lemma:postmers:i_0} and Proposition~\ref{prop:postmers:general_case} for postmers. To compute $\pi_k^\gamma(w)$, we only need to fill that table up to $\alpha=k-m$ for antemers, and up to $\beta = m+\beta_{\max}$ for postmers. We must also, of course, compute all the derived quantities $\mathbf{R}, \mathbf{T}_i$, and so on.

\begin{theorem}\label{th:complexity}
For any $w\in\Sigma^m$ and any $\gamma\in\Sigma^m$, $\pi_k^\gamma(w)$ can be computed in $O(|\Sigma|\cdot km^2)$ time and $O(km)$ space.
\end{theorem}
\begin{proof}
The proof is deferred to Appendix~\ref{app:th:complexity}.
\end{proof}

The Python code for computing $\pi_k^\gamma(w)$, as well as the scripts for reproducing the figures of this paper, is available at \url{https://github.com/fingels/minimizer_counting_function}. In Figure~\ref{fig:numerical results}, we computed $\pi_k^\gamma(w)$ for all possible $w\in\lbrace \ch{A}, \ch{C},\ch{G},\ch{T}\rbrace^m$ ($m=10, k=31$), and for several keys $\gamma$, including the aforementioned orders of the literature: lexicographical, anti-lexicographical \cite[Def.~8.4]{groot2025optimal} and alternating \cite{roberts2004reducing}, plus some random orders. As one can see, the first character of $\gamma$ drives where the peak of the distribution are found, and the subsequent characters drive second-order, so to speak, behaviour of the distribution. If one sorts the buckets by decreasing size (as shown in Appendix~\ref{app:sorting_distributions}), all keys lead to a somewhat identical repartition, allowing us to conclude that the key really does play a role in ``shuffling'' buckets. As a result, no key seems better than another in the sense of producing more balanced distributions, but since the empirical biological data are not uniform, one can still hope to build keys that avoid the \mmer{s} that are too frequent in the data.

\begin{figure}[h!]
    \centering
    \includegraphics[width=\textwidth]{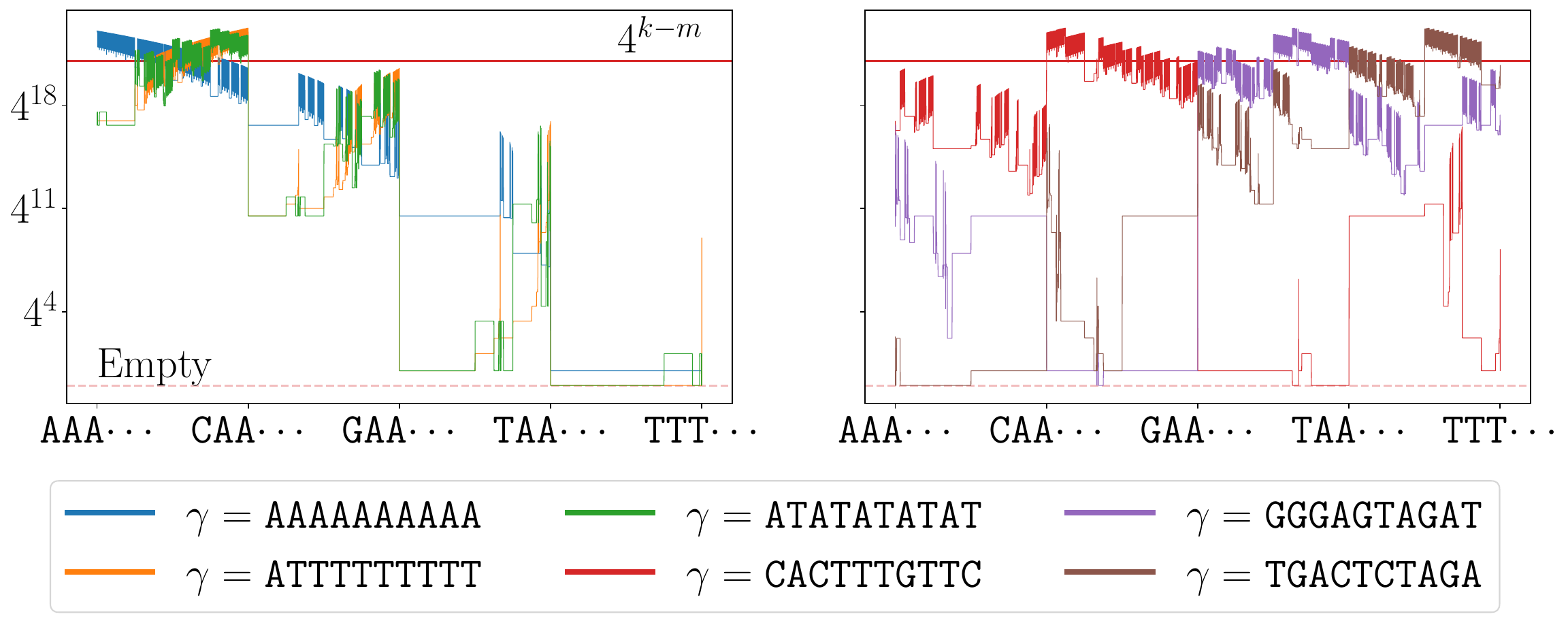}
    \caption{$w\in\lbrace \ch{A}, \ch{C},\ch{G},\ch{T}\rbrace^m\mapsto \pi_k^\gamma(w)$ with $k=31$, $m=10$, for several values of $\gamma$: (left) $\ch{A}\cdots\ch{A}$ is the standard lexicographical order, $\ch{A}\ch{T}\cdots\ch{T}$ is the anti-lexicographical order and $\ch{ATAT}\cdots $ is the alternating order; (right) three random keys starting by \ch{C},\ch{G} and \ch{T}. The solid red horizontal line represents a perfectly balanced partition; the dotted horizontal red line corresponds to empty buckets.}
    \label{fig:numerical results}
\end{figure}

Keeping in mind the possible applications mentioned in the introduction --- and ignoring the factor $|\Sigma|$, time complexity in $O(km^2)$ (recall Theorem~\ref{th:complexity}) may still seem quite expensive, especially if we consider that one might want to calculate $\pi_k^\gamma(w)$ for a large number of \mmer{s} $w$ or keys $\gamma$ --- thousands, millions, or even all of them (provided that $m$ is small enough). In particular, it may be more useful to obtain an approximation of the value $\pi_k^\gamma(w)$ faster than to compute the exact value. This option is discussed in Appendix~\ref{app:approximation}.

\section{Conclusion and future work}
In this work, we extented the method of \cite{ingels2024number} to compute $\pi_k^\gamma(w)$ for any key $\gamma$, whereas \cite{ingels2024number} only allowed $\gamma=\ch{A}\cdots \ch{A}$, i.e. corresponding to the lexicographical order. Therefore, our method allow to access the counting function of $\Sigma^m$ orders on \mmer{s}. We believe that our method could be straightforwardly extended to a broader range of orders on \mmer{s}, that would be of the following form.

\begin{conjecture}
Let $g : \Sigma^m\to\Sigma^m$ be such that $g(a_1\cdots a_m) = g_1(a_1)\cdots g_m(a_m)$ where $g_i:\Sigma\to\Sigma$. Then, by defining the minimizer $\min_g(x)$ of a \kmer $x$ as its leftmost \mmer $w$ that minimizes $\texttt{rank}(g(w))$ (for the lexicographical order), we can compute
$\pi_k^g(w) = \left\lbrace x\in\Sigma^k : \min_g(x) = w\right\rbrace$
in $O(|\Sigma|\cdot km^2)$ time and $O(km)$ space by adapting the equations presented in this article.
\end{conjecture}
Indeed, since we fill letter by letter the antemers and postmers, any function $g$ that would also operate letter by letter (such as $g_\gamma(w)=w\oplus\gamma$ in our case) would be compatible with our method. The main difference would be that, since $g_i$ would not be an involution in the general case, we would have to carefully handle cases of equality of the form $g_i(b) = g_i(a_i)$ where $b$ is the letter to determine and $a_i$ the $i$-th letter of $w$. Would this conjecture holds, we would go from covering $|\Sigma|^{m}$ to $|\Sigma|^{m\cdot |\Sigma|}$ orders on \mmer{s} among the $(|\Sigma|^m)!$ possible. To solve more complex cases (e.g., where the $i$-th letter of $g(w)$ depends on $\geq 2$ letters), we conjecture that a completely different approach from the one presented here will have to be developed.

Finally, in future work, we plan to exploit the function $\pi_k^\gamma(w)$ for the bioinformatics applications mentioned in the introduction, in particular by studying how to produce a $\gamma$ (equivalently, an order on \mmer{s}) that is suited to a particular dataset and application. In light of the recent idea of using several minimizers/orders on \mmer{s} to improve minimizers-based sampling \cite{ingels2025minimizer}, we wish also to consider using a set of keys $\gamma_1,\dots,\gamma_N$ to determine the best (in a sense to be determined) minimizer among $N$ candidates, e.g. to develop new heuristics for \kmer partitioning.

%

\begin{credits}

\subsubsection{\ackname} This work was supported by the French National Research Agency  full-RNA [ANR-22-CE45-0007]. With financial support from ITMO Cancer of Aviesan within the framework of the 2021-2030
Cancer Control Strategy, on funds administered by Inserm.

\subsubsection{\discintname}
The authors have no competing interests to declare that are
relevant to the content of this article.
\end{credits}

\bibliographystyle{splncs04}
\bibliography{biblio}

\appendix

\section{Proof of Theorem~\ref{th:complexity}}\label{app:th:complexity}

We start by the following lemmas.

\begin{lemma}\label{lemma:preprocess}
Computing the quantities $\mathbf{R}$, $i_{\max}$, $\beta_{\max}$, $\Sigma_i$, $\Sigma_A(i)$, $\Sigma_A^{=0}(i)$, $\Sigma_A^{\neq 0}(i)$, $\Sigma_D^{=0}(i)$, $\Sigma_D^{\neq 0}(i)$, $\mathbf{S}_l$ and $\mathbf{T}_i$ can be done in a preprocessing step taking $O(m^2)$ space and $O(m^2\cdot|\Sigma|)$ time.
\end{lemma}
\begin{proof}
The proof is deferred to Appendix~\ref{app:lemma_preprocess}.
\end{proof}

\begin{lemma}\label{lemma:antemers_computation}
Let $\alpha\geq 0$. Computing all values $A_i(\alpha')$, $0\leq i\leq i_{\max}$ and $A(\alpha')$ for $0\leq \alpha'\leq \alpha$ can be done in $O( i_{\max} \cdot \alpha)$ space and $O(|\Sigma|\cdot i_{\max}^2 \cdot \alpha)$ time.
\end{lemma}
\begin{proof}
The proof is deferred to Appendix~\ref{app:antemers_and_postmers_computation}.
\end{proof}

\begin{lemma}\label{lemma:postmers_computation}
Let $\beta\geq 0$. Computing all values $P_i(\beta')$, $0\leq i\leq m$ and $P(\beta')$ for $0\leq \beta'\leq \beta$ can be done in $O(m\cdot\beta )$ space and $O(|\Sigma|\cdot m^2\cdot \beta)$ time.
\end{lemma}
\begin{proof}
The proof is deferred to Appendix~\ref{app:antemers_and_postmers_computation}.
\end{proof}

To compute $\pi_k^\gamma(w)$ using Proposition~\ref{prop:reformulation}, we take $\alpha = k-m$ and $\beta= m+\beta_{\max}$. Notice that $i_{\max}\leq m-1$ and $\beta_{\max}\leq k-m$, hence a worst case complexity of $O\big(|\Sigma|\cdot km^2\big)$ time and $O\big(km\big)$ space, as claimed.

\subsection{Proof of Lemma~\ref{lemma:preprocess}}\label{app:lemma_preprocess}

Assuming that computing $w'\oplus\gamma$ can be done in constant time, computing $\mathbf{R}$ cost $O(m^2)$ both in space and time. $i_{\max}$, $\beta_{\max}$ can be updated on the fly (as soon as a $<$ symbol is found) while computing $\mathbf{R}_{i,j}$ at no extra cost. All alphabets $\Sigma_i$, $\Sigma_A(i), \Sigma_A^{=0}(i), \Sigma_A^{\neq 0}(i), \Sigma_D^{=0}(i)$ and $\Sigma_D^{\neq 0}(i)$ can be computed in at most $O(m^2\cdot|\Sigma|)$ total time (for all values of $i$) and $O(|\Sigma|)$ space, assuming in place intersection. Computing $\mathbf{S}_l$, from Lemma~\ref{lemma:antemer_i=alpha}, takes $O(m)$ in time and space. $\mathbf{T}_i$ takes $O(m\cdot|\Sigma|)$ space and time. All in all, the preprocessing step takes $O(m^2)$ space and $O(m^2\cdot|\Sigma|)$ time.

\subsection{Proof of Lemmas~\ref{lemma:antemers_computation} and~\ref{lemma:postmers_computation}}\label{app:antemers_and_postmers_computation}

\paragraph{Antemers} To compute $A_i(\alpha'), A(\alpha')$ for $0\leq i\leq i_{\max}$ and $0\leq \alpha'\leq \alpha$, one must fill --- in a dynamic programming fashion --- a table of dimensions $(i_{\max}+1)\times (\alpha+1)$. Hence, $O(\alpha \cdot i_{\max})$ space is needed, assuming storing integers takes $O(1)$ space\footnote{In practice, the integers we store are necessarily $\leq |\Sigma|^k$, so they require as much place as storing a \kmer. As an example, $31$-mers are usually encoded with $64$ bits.}. $A_0$ is computed in $O(1)$ from Lemma~\ref{lemma:antemer_i=0} whereas $A_i(i)$ is computed in $O(i)$ from Lemma~\ref{lemma:antemer_i=alpha}. For the general case, from Proposition~\ref{prop:antemer:general_case}, for $0<i\leq i_{\max}$ $A_i$ is computed in time proportional to $\sum_{a\in\Sigma_A^{\neq 0}(i)}  (i_{\max} - (i-\mathbf{T}_i(a)+2))$. Since $\mathbf{T}_i(a)\leq i+1$, we end up after simplification to $O(|\Sigma|\cdot i_{\max})$ time to compute $A_i$. Computing $A$ takes $O(i_{\max})$ time, hence a total time of $O(|\Sigma|\cdot i_{\max}^2 \cdot \alpha)$.

\paragraph{Postmers} The line of reasoning is quite similar here, except that the quantities $\widetilde{\mathbf{T}}_i(\cdot,\beta)$, $\Sigma_P^{=0}(i,\beta)$ and $\Sigma_P^{\neq 0}(i,\beta)$ have not been precomputed since they depend on $\beta$. To compute all required values, we need a table of dimensions $(m+1)\times (\beta+1)$, hence $O(m\cdot \beta)$ space is required.

From Lemma~\ref{lemma:postmers:small_values}, computing $P_i(\beta')$, $\beta'< m$, takes time proportional to $\sum_{a\in\Sigma_D^{\neq0}(i)}\big((\beta'-\mathbf{T}_i(a)+1) - (i-\mathbf{T}_i(a)+2)\big) = (\beta'-i-1)\cdot|\Sigma_D^{=0}(i)| =O(\beta'\cdot |\Sigma|)$, hence a total time complexity of $O(m^3\cdot|\Sigma|)$ to compute the first $m$ columns of the table. For $\beta'=m$, computing the $(m+1)$-th column takes $O(m)$ since computing $P_i(m)$ is $O(1)$, per Lemma~\ref{lemma:postmers:m}. For $\beta'>m$, computing $P_0(\beta')$ is also $O(1)$ --- Lemma~\ref{lemma:postmers:i_0}.

For the general case $\beta'>m$, $0<i\leq m$, to compute $P_i(\beta')$ using Proposition~\ref{prop:postmers:general_case}, we first need to compute the quantities $\widetilde{\mathbf{T}}_i(\cdot,\beta')$, $\Sigma_P^{=0}(i,\beta')$ and $\Sigma_P^{\neq 0}(i,\beta)$. $\widetilde{\mathbf{T}}_i(a,\beta)$ takes $O(1)$ per letter $a$ --- see \eqref{eq:tilde_prefix}, hence a total of $O(|\Sigma|)$. Computing $\Sigma_P(i,\beta')$ (from Lemma~\ref{lemma:postmers:letters_possible}) takes at most $O(m \cdot |\Sigma|)$ time and $O(|\Sigma|)$ space, with in place operations. $\Sigma_P^{=0}(i,\beta)$ and $\Sigma_P^{\neq 0}(i,\beta)$ can be computed afterwards in one $O(\Sigma|)$ time pass. In total, for all values of $i$, computing all those quantities takes total time $O(m^2\cdot |\Sigma|\cdot (\beta-m))$ and total space $O(|\Sigma|)$ (since the alphabets can be discarded after being used). Finally, computing $P_i(\beta')$ takes time proportional to $\sum_{a\in \Sigma_P^{\neq 0}(i,\beta)} \big(m -  (i -\widetilde{\mathbf{T}}_i(a,\beta)+2)\big)$. Since $\widetilde{\mathbf{T}}_i(a,\beta)\leq i+1$, it costs in the end $O(|\Sigma|\cdot m)$ per value, hence $O(|\Sigma|\cdot m^2\cdot(\beta-m))$ for all $P_i(\beta')$ values with $i>0$ and $m<\beta'<\beta$.

Combining everything together, the final time complexity for the whole table is $O(|\Sigma|\cdot m^2\cdot \beta)$.

\section{Sorting the distributions by decreasing size of buckets}\label{app:sorting_distributions}

If we take the distributions of Figure~\ref{fig:numerical results}, and sort the buckets by decreasing size, we obtain Figure~\ref{fig:distribution_sorted}. As one can see, the most filled buckets have similar sizes no matter the keys, and only the middle and tail of the distributions really differ.

\begin{figure}[h]
    \centering
          \includegraphics[width=\textwidth]{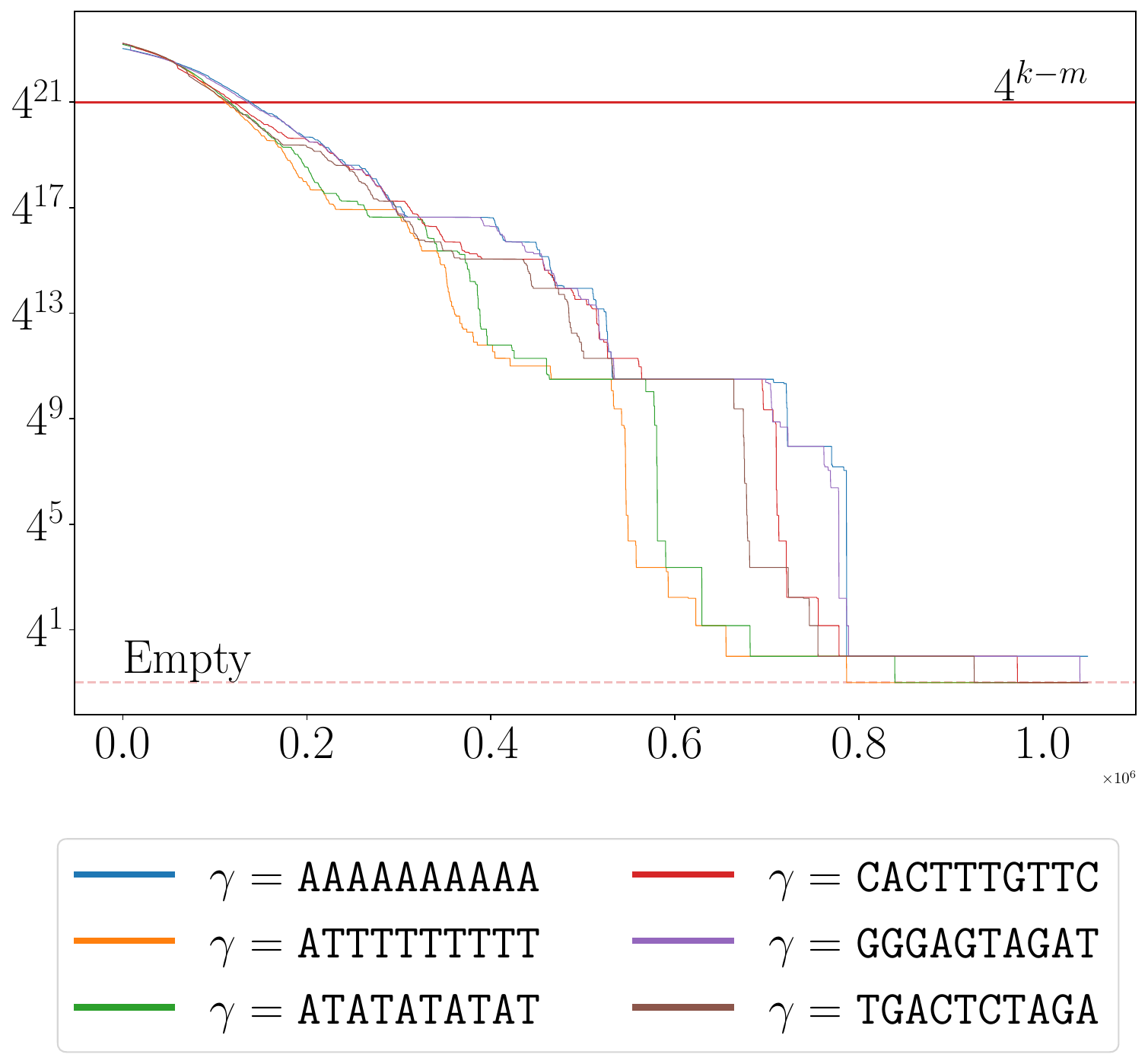}  
      \caption{Distributions of Figure~\ref{fig:numerical results}, sorted by decreasing size of the buckets. As for Figure~\ref{fig:numerical results}, $k=31$, $m=10$.}
    \label{fig:distribution_sorted}   
\end{figure}

\section{Approximating $\pi_k^\gamma$}\label{app:approximation}

As evoked in Section~\ref{sec:computation}, one might want to speed up the computation time for $\pi_k^\gamma(w)$ --- done in $O(km^2)$ time and $O(km)$ space, as per Theorem~\ref{th:complexity} --- at the expanse of obtaining only an approximation of the value, that might be enough for potential applications, e.g. for comparing different vigemin.

With the exception of a few pathological cases (which can be detected quickly), most choices of $w$ and $\gamma$ will lead to asymptotic behaviour of the form $$\pi_k^\gamma(w) \underset{k\to\infty}{\sim} \alpha\cdot \rho^k$$ where $\alpha,\rho$ are some constants depending on $w,\gamma$; see for instance \cite[Th.~VI.7 and Th.~IV.9]{flajolet2009analytic}. What we propose to do is, for fixed $w,\gamma$, to compute $\pi_k^\gamma(w)$ for a few distinct values $k_1,\dots,k_n$ --- leading to values $\pi_1,\dots,\pi_n$ --- and do a simple linear regression of the form
$$\log\pi_i = ak_i +b + \epsilon_i$$
for some constants $a,b$ and error $\epsilon_i$. $a$ and $b$ can be estimated in $O(n)$ via least-square method, and then we predict $\log\pi_k^\gamma(w)$ for large values of $k$ in $O(1)$ as $\widehat{\log\pi_k^\gamma(w)}= \hat{a}k+\hat{b}$. If the desired application only requires comparing values $\pi_k^{\gamma}(\cdot)$ with each other, we can compare the predicted logarithms directly, without applying the exponential function, in order to avoid blowing up the error. The total computation time goes from $O(km^2)$ per value to  $O(m^2k_n+n)=O(m^2k_n)$ since $n\leq k_n$. If $k_n\ll k$, one can save a lot of time when iterating over many choices of $w,\gamma$.

In Figure~\ref{fig:approximation}, with $m=10$, we used\footnote{We start at $k_1=15$ instead of $k_1=10$ to allow the values to enter a linear regime, as our experiments have shown that the first few values do not fit nicely on a line.} $k_1=15, k_2=16,\dots,k_{10}=25$, computed $\hat{a},\hat{b}$ for $1000$ random \mmer $w\in\lbrace \ch{A}, \ch{C},\ch{G},\ch{T}\rbrace^m$ and $4$ random $\gamma\in\lbrace \ch{A}, \ch{C},\ch{G},\ch{T}\rbrace^m$, estimated $\widehat{\log\pi_{100}^\gamma(w)}$ and compared to the actual value $\log\pi_k^\gamma(w)$. As one can see, for most of the values, the prediction is somewhat close to the actual value (with a tendency to overestimate). Also, the key does not seems to influence the estimator.

\begin{figure}[h!]
    \centering
    \includegraphics[width=\textwidth]{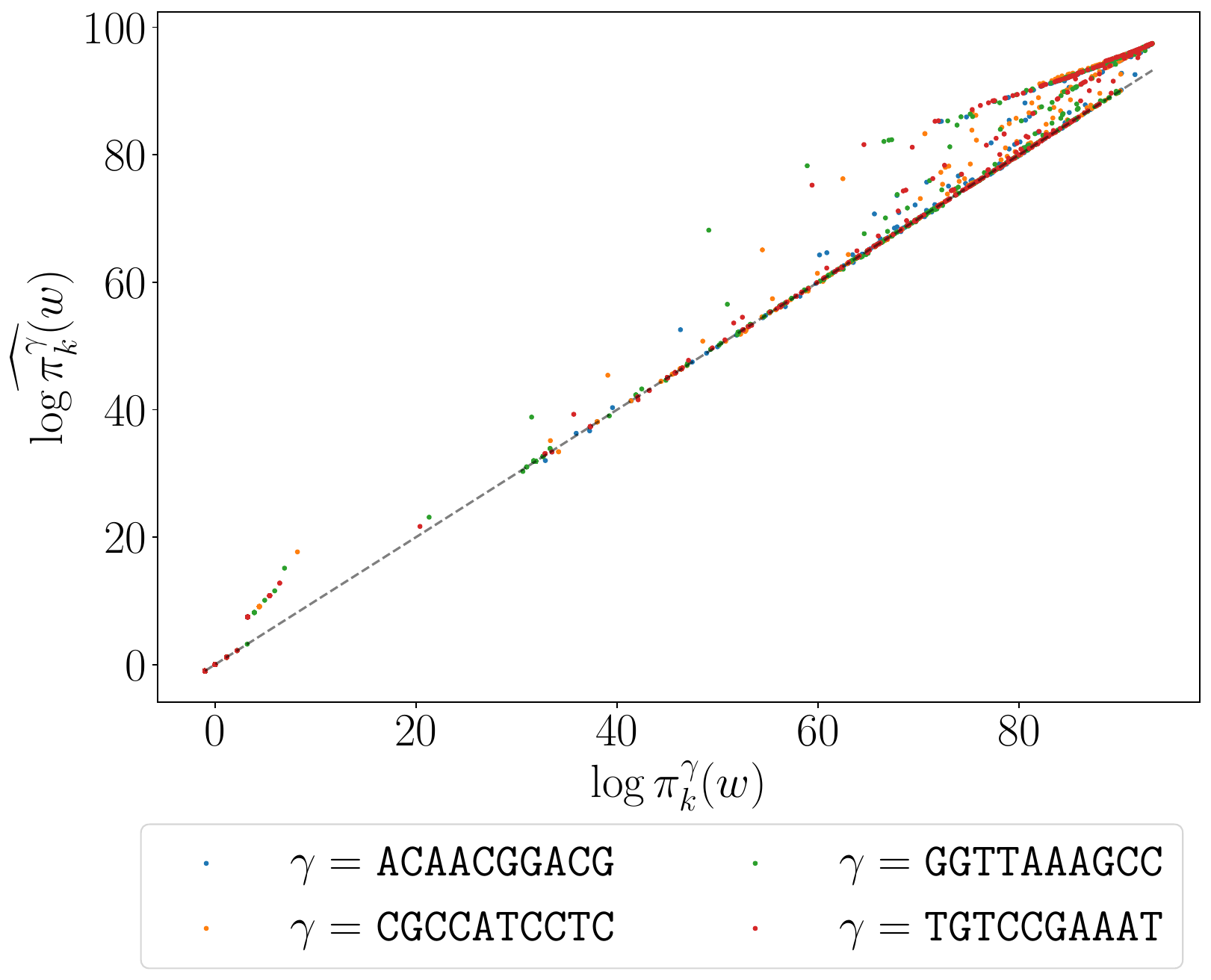}
    \caption{Theoretical values $\log \pi_k^\gamma(w)$ vs predicted values $\widehat{\log\pi_k^\gamma(w)}$ for $1000$ random $w\in\lbrace \ch{A}, \ch{C},\ch{G},\ch{T}\rbrace^m$ and $4$ random keys $\gamma\in\lbrace \ch{A}, \ch{C},\ch{G},\ch{T}\rbrace^m$, with $m=10$, $k=100$ and the estimator for $\widehat{\log\pi_k^\gamma(w)}$ is computed on $k_1=15,k_2=16,\dots,k_{10}=25$. Black dashed line corresponds to $y=x$.}
    \label{fig:approximation}
\end{figure}



\end{document}